\documentclass[letterpaper, 10 pt, conference]{ieeeconf}

\usepackage{amsmath}

\usepackage{amsthm}
\usepackage{amsfonts}
\usepackage{bbm} 
\usepackage{graphicx}
\usepackage{subfigure}
\usepackage{cite}
\usepackage{textcomp}
\usepackage{url}
\usepackage{ amssymb }
\usepackage{color}

\newtheorem{defn}{Definition}
\newtheorem{thm}{Theorem}

\newtheorem{prop}{Proposition}
\newtheorem{cor}{Corollary}

\newtheorem{notation}{Notation}

\newtheorem{exam}{Example}
\newtheorem{assump}{Assumption}

\newcommand{\E}{\mathcal{E}}
\newcommand{\EE}{\mathbb{E}}
\newcommand{\V}{\mathbb{V}}
\newcommand{\G}{\mathcal{G}}
\newcommand{\Aut}{\mathrm{Aut}}
\newcommand{\HH}{\mathcal{H}}

\begin{document}

\title{Symmetry-Induced Clustering in Multi-Agent Systems using Network Optimization and Passivity}
\author{Miel Sharf and Daniel Zelazo
\thanks{M. Sharf and D. Zelazo are with the Faculty of Aerospace Engineering, Israel Institute of Technology, Haifa, Israel.
    {\tt\small msharf@tx.technion.ac.il, dzelazo@technion.ac.il}.  This work was supported by the German-Israeli Foundation for Scientific Research and Development.}
}

\maketitle
\begin{abstract}
This work studies the effects of a weak notion of symmetry on diffusively-coupled multi-agent systems. We focus on networks comprised of agents and controllers which are maximally equilibrium independent passive, and show that these converge to a clustered steady-state, with clusters corresponding to certain symmetries of the system. Namely, clusters are computed using the notion of the exchangeability graph. We then discuss homogeneous networks and the cluster synthesis problem, namely finding a graph and homogeneous controllers forcing the agents to cluster at prescribed values. 
\end{abstract}

\section{Introduction}\label{sec.Intro}
Multi-agent systems have been in the limelight of control research for the last decade due to their many applications in various areas of science and engineering, e.g. neuroscience \cite{Schnitzler2005}, biochemical systems \cite{Scardovi2010}, and robotics \cite{Chopra2006}.
Different approaches have been proposed for establishing a unified theory for multi-agent systems, and one of the most prominent schools of thought is modeling using tools from passivity theory \cite{Bai2011}. Studying multi-agent systems using passivity was first proposed by Arcak in \cite{Arcak2007}, but since then many variations have been explored, including incremental passivity \cite{Pavlov2008}, relaxed co-coercivity \cite{Stan2007}, and various notions of equilibrium-independent passivity frameworks \cite{Hines2011,Burger2014,Sharf2018a}.  

One important problem in the theory of multi-agent systems is the consensus problem. The consensus problem fixes a collection of agents, and tasks one to design a distributed control law forcing all agents to converge to the same output. The consensus problem has applications in almost all areas of multi-agent systems, including distributed computation \cite{Xiao2004}, robotics \cite{Chopra2006}, biochemical systems \cite{Scardovi2010}, etc. A generalization of the concensus problem is the clustering problem, in which the agents are divided into different groups (namely, clusters). The problem then tasks one to design a distributed control law forcing agents in the same cluster to synchronize, while agents in different clusters do not synchronize. The clustering problem is essential in fields like ecology \cite{Lewi2002}, neuroscience \cite{Schnitzler2005}, and biomimicry of swarms \cite{Passino2002}. Various methods have been used to study clustering, e.g. sturctural balance of the underlying graph \cite{Altafini2013}, pinning control \cite{Qin2013} and inter-cluster nonidentical inputs \cite{Han2013}.

We approach the clustering problem using symmetry. The notion of symmetry is one of the cornerstones of mathematics and physics. It is used in control theory extensively for many different applications. Examples include designing observers \cite{Bonnabel2009}, more efficient algorithms for model-predictive control \cite{Danielson2015}, and bipedal locomotion \cite{Spong2005}. In cooperative control, symmetry on the network level was used in \cite{Rahmani2007,Chapman2014,Chapman2015} to study controllability and observability. However, these works discuss network symmetries preserving the agents' models, which can be different even if the agents are equivalent. Moreover, the current literature about symmetries in multi-agent systems deals with symmetries in the trajectories of the agents, although consensus and clustering only require symmetries on the steady-state level. Our contributions are detailed below:
\begin{itemize}
\item We define the notion of a weak equivalence of systems, allowing us to define a more general, model-free notion of network symmetries on the steady-state level. We also define the weak automorphism group of a multi-agent system, and show that the set of steady-states is invariant under it.
\item We use the notion of maximal equilibrium independent passivity (MEIP), developed in \cite{Burger2014}, to show that under a proper passivity assumption, the output of the network must converge to a clustered steady-state, in which the clusters can be predicted using network symmetries.
\item Lastly, we discuss the case of homogeneous systems and the problem of cluster synthesis. We demonstrate a solution to the problem for a specific case. 
\end{itemize}   

The rest of this paper is organized as follows. Section \ref{sec.NetOpt} provides background on passivity and network optimization. Section \ref{sec.Symmetry} presents the paper's main results, followed by Section \ref{sec.Simulations} presenting two examples of the presented theory.

\paragraph*{Notations}
This work employs basic notions from algebraic graph theory \cite{Godsil2001}.  An undirected graph $\mathcal{G}=(\mathbb{V},\mathbb{E})$ consists of a finite set of vertices $\mathbb{V}$ and edges $\mathbb{E} \subset \mathbb{V} \times \mathbb{V}$.  We denote by $k=\{i,j\} \in \mathbb{E}$ the edge that has ends $i$ and $j$ in $\mathbb{V}$. For each edge $k$, we pick an arbitrary orientation and denote $k=(i,j)$ when $i\in \mathbb{V}$ is the \emph{head} of edge $k$ and $j \in \mathbb{V}$ the \emph{tail}.  The incidence matrix of $\mathcal{G}$, denoted $\mathcal{E}\in\mathbb{R}^{|\mathbb{E}|\times|\mathbb{V}|}$, is defined such that for edge $k=(i,j)\in \mathbb{E}$, $[\mathcal{E}]_{ik} =+1$, $[\mathcal{E}]_{jk} =-1$, and $[\mathcal{E}]_{\ell k} =0$ for $\ell \neq i,j$. For a graph $\G$, an automorphism of $\G$ is a permutation $\psi:\V\to\V$ such that $i$ is connected to $j$ if and only if $\psi(i)$ is connected to $\psi(j)$. We denote its automorphism group by $\Aut(\G)$. 

\section{Background: Passivity and Network Optimization for Multi Agent Systems}\label{sec.NetOpt}
The role of network optimization theory in cooperative control was introduced in \cite{Burger2014}, and was later developed in \cite{Sharf2017,Sharf2018a}. This section summarizes the main results of \cite{Burger2014}. 

\subsection{Maximally Monotone Dynamical Systems}
We consider SISO dynamical systems of the form:
\begin{align} \label{Upsilon}
\Upsilon :
\dot{x} = f(x,u), \;
y = h(x,u),
\end{align}
where $u \in \mathbb{R}$ is the input and $y \in \mathbb{R}$ is the output. Many variants of passivity were studied for such systems. We focus on two - \emph{equilibrium independent passivity} (EIP) and \emph{maximally equilibrium independent passivity} (MEIP). EIP was first introduced in \cite{Hines2011}. It requires the existence of an equilibrium input-output map, mapping constant steady-state inputs to constant steady-state outputs. It also requires passivity with respect to said input-output pairs. 

Another, more general, notion was introduced in \cite{Burger2014}. It was coined \emph{maximal equilibrium independent passive} (MEIP). As EIP, it requires passivity with respect to any steady-state input-output pair. Unlike EIP, it considers the collection of all pairs $(u_{ss},y_{ss})$ of steady-state inputs and outputs, denoted by $k_\Upsilon$. It gives rise to two set-valued maps - if $\mathrm u$ is a steady-state input and $\mathrm y$ is a steady-state output, we denote the steady-state outputs associated with $\mathrm u$ by $k_\Upsilon(\mathrm u)$, and the steady-state inputs associated with $\mathrm y$ by $k_\Upsilon^{-1}(\mathrm y)$. The image of these set valued maps can have more than one point, or no points at all. For example, if $\Upsilon$ is the single integrator $\dot{x}=u,\: y=x$, then $k_\Upsilon = \{(0,y):\: y\in\mathbb{R}\}$.  This is the main difference between EIP and MEIP.  In EIP,  the steady-state input-output maps are \emph{functions}, while for MEIP they are \emph{relations}

\begin{defn}[\small{Maximal Equilibrium Independent Passivity \cite{Burger2014}}]
Let $\Upsilon$ be as in \eqref{Upsilon}. The system $\Upsilon$ is \emph{maximally equilibrium independent monotonic (output-strictly) passive} if:
\begin{enumerate}
\item[i)] The system $\Upsilon$ is (output-strictly) passive with respect to any steady state $(u_{ss},y_{ss})\in k_\Upsilon$ \cite{Khalil2001}.
\item[ii)] The relation $k_\Upsilon$ is \emph{maximally monotone}, i.e., if $(u_1,y_1)(u_2,y_2)\in k_\Upsilon$ then \small$(u_1-u_2)(y_1-y_2)\ge 0$,\normalsize and $k_\Upsilon$ is not contained in a larger monotone relation \cite{Rockafeller1997}.
\end{enumerate}
\end{defn}
Such systems include (among others) single integrators, gradient systems, Hamiltonian systems on graphs, and others (see \cite{Burger2014,Sharf2017,Sharf2018a,Sharf2018b} for more examples). 

The interest in monotone relations stems from their connection to convex functions. A theorem by Rockafellar \cite{Rockafellar1966} states that maximal monotone relations are given by the subdifferential of a convex function $\mathbb{R}\rightarrow\mathbb{R}$, and vice versa. Furthermore, this correspondence is unique up to a constant added to the convex function. In particular, for MEIP systems, there is some convex function $K_\Upsilon$ such that the steady-stat relation $k_\Upsilon(u)$ is the subgradient $\partial K_\Upsilon(u)$. In \cite{Burger2014,Sharf2018a} this property was used to build a network optimization-based framework to find steady-states.
\subsection{Diffusively Coupled Networks}
In this subsection, we describe the structure of the network dynamical system studied in \cite{Burger2014}. We also present the connection between networked dynamical systems and network optimization theory.

Consider a collection of agents interacting over a network $\mathcal{G}=(\mathbb{V},\mathbb{E})$.  The nodes $i\in \V$ are assigned dynamical systems $\Sigma_i$, and the edges $e\in\mathbb{E}$ are assigned controllers $\Pi_e$, having the following form:
\begin{align} \label{Dynamics}
\Sigma_i: 
\begin{cases}
\dot{x}_i = f_i(x_i,u_i) \\
y_i = h_i(x_i,u_i),
\end{cases}
\Pi_e: 
\begin{cases}
\dot{\eta}_e= \phi_e(\eta_e,\zeta_e) \\
\mu_e = \psi_e(\eta_e,\zeta_e)
\end{cases}.
\end{align}

We consider stacked vectors of the form $u=[u_1^T,...,u_{|\mathbb{V}|}^T]^T$ and similarly for $y,\zeta$ and $\mu$.  
The network system is diffusively coupled with the controller input described by $\zeta = \mathcal{E}^Ty$, and the control input to each system by $u = -\mathcal{E}\mu$, where $\E$ is an incidence matrix of the graph $\G$.  This structure is denoted by the triplet $(\G,\Sigma,\Pi)$, and is illustrated in Fig. \ref{ClosedLoop}. For the rest of this paper, we will assume one of the following two alternatives.  If this is not the case, see \cite{Jain2018} and \cite{Jain2019} for plant augmentation techniques.

\begin{assump}\label{assump.1}
The agents $\Sigma_i$ are output-striclty MEIP and the controllers $\Pi_e$ are MEIP.
\end{assump}

\begin{assump}\label{assump.2}
The agents $\Sigma_i$ are MEIP and the controllers $\Pi_e$ are output-strictly MEIP.
\end{assump}

\begin{figure} [!t] 
    \centering
    \includegraphics[scale=0.65]{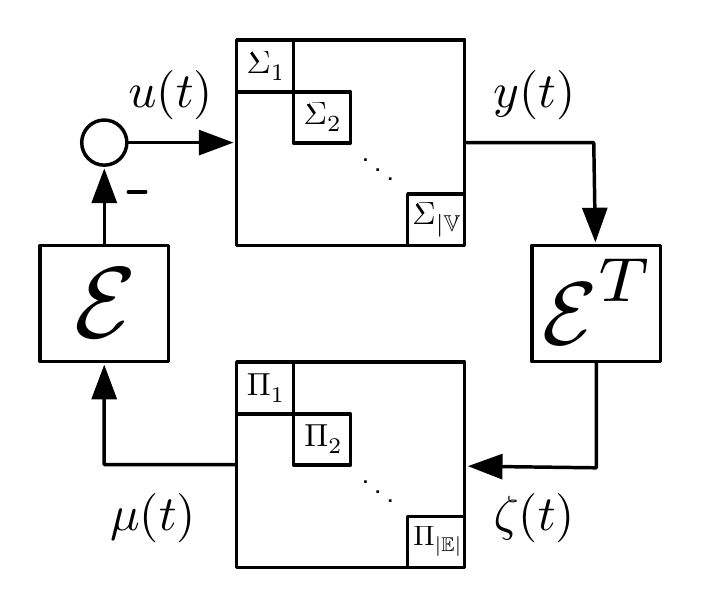}
    
    \caption{Block-diagram of the closed loop.}
    \label{ClosedLoop}
    
\end{figure}

We denote the steady-state input-output relations of the node $i$ and the edge $e$ by $k_i$ and $\gamma_e$, respectively. 
Owing to Rockafellar's result, we take convex functions $K_i(u_i)$ and $\Gamma_k(\zeta_k)$  such that $\partial K_i(u_i) = k_i$ and $\partial \Gamma_k(\zeta_k) = \gamma_k$. We consider the stacked relations $k(u)$ and $\gamma(\zeta)$ by concatenating the $k_i(u_i)$'s and $\gamma_k(\zeta_k)$'s respectively. We also define the convex functions $K(u)=\sum_{i\in\mathbb{V}} K_i(u_i)$ and $\Gamma(\zeta) = \sum_{k\in\mathbb{E}} \Gamma_k(\zeta_k)$. It is straightforward to check that $\partial K(u)=k(u)$ and $\partial \Gamma(\zeta) = \gamma(\zeta)$.

In order to state the main theorem, we introduce the notion of the dual function. The dual function of $K$ is defined by $K^\star(y) = \min_{u}\{y^Tu - K(u)\}$ \cite{Rockafeller1997}. It is also a convex function, and it possesses the property that $\partial K^\star(y) = k^{-1}(y)$. One can similarly define the convex dual $\Gamma^\star(\mu)$ of $\Gamma$. 
We are now ready to state the main result from \cite{Burger2014}.
\begin{thm}[\cite{Burger2014,Sharf2017}]\label{thm.Analysis}
Consider the diffusively-coupled system $(\G,\Sigma,\Pi)$, and assume either Assumption \ref{assump.1} or \ref{assump.2} holds. Then the signals $u(t),y(t),\zeta(t),\mu(t)$ of the closed-loop system converge to some steady-state values $\hat{u},\hat{y},\hat{\zeta},\hat{\mu}$. These values are the (primal-dual) solutions of the following pair of convex optimization problems:
\begin{center}
\begin{tabular}{ c||c }
 \textbf{Optimal Potential Problem} & \textbf{Optimal Flow Problem}  \\ \hline   \\
 $ \begin{array}{cl} \underset{y,\zeta}{\min} &K^\star(y) + \Gamma(\zeta)\\
s.t.&\mathcal{E}^Ty = \zeta 
\end{array} $&  $ \begin{array}{cl}\underset{u,\mu}{\min}& K(u) + \Gamma^\star(\mu) \\
s.t. &\mu = -\mathcal{E}u.
\end{array} $ 
\end{tabular}
\end{center}
\end{thm}

Lastly, in some cases, we make the following assumption:
\begin{assump}\label{assump.3}
The equality $\gamma_e(-x) = -\gamma_e(x)$ holds for every $x$ and every $e\in \EE$.
\end{assump}
This assumption implies that the steady-state of the closed-loop system does not depend on the choice of orientation.  Indeed, it implies that $\Gamma_e(x) = \Gamma_e(-x)$, and the result now follows from Theorem \ref{thm.Analysis}. For example, this assumption holds of all of the controllers are static nonlinearities which are odd functions. As will be shown later, this assumption helps us achieve consensus, but bars us from achieving clustering.
\vspace{-2pt}
\section{Symmetry and Clustering\\ in Multi-Agent Systems}\label{sec.Symmetry}
In this section, we develop the theory for symmetries in Multi-agent system. We begin with a brief overview about the role of symmetry in control and in multi-agent systems, 
\subsection{The Weak Automorphism Group of a Multi-Agent System}
As stated in the introduction, symmetries have been used in the study of control laws for many systems \cite{Bonnabel2009,Danielson2015,Spong2005}. In cooperative control, symmetries were used in the study of controllability and observability \cite{Rahmani2007,Chapman2014,Chapman2015}. Namely, in \cite{Rahmani2007}, it is shown that if we have a weighted graph $\G=(\mathbb{W,E,V})$ and input nodes $S \subset \mathbb{V}$, then the controlled consensus system $\dot{x} = -L(\G)x + Bu$, where $L(\G)$ is the graph Laplacian and $B$ is supported on $S$, is uncontrollable, as long as there exists a nontrivial graph automorphism $\psi\in \Aut(\G)$ such that $P_\psi$ commutes with $L(\G)$ and $P_\psi B = B$. Later, \cite{Chapman2014} expended this idea to ``Fractional Automorphism", using the inherent linearity of the system.

Pushing this idea a step further, we want to consider more general systems. A first step is the case of linear systems. If we try and mimic \cite{Chapman2014}, then we require that the symmetry matrix $P_\psi$, which corresponds to some permutation, commutes with the dynamics matrix $A$ of the entire system. This has a few drawbacks - The main one is that this is extremely model-dependent, i.e., different matrices $A$ might yield different symmetries, even though the agents are equivalent. Specifically, on a two-vertex graph with one edge, where both agents have the same dynamics, but different realizations of the model, the graph automorphism exchanging the vertices is not a symmetry.

One possible direction to remove this problem is to try and use model-independent sizes, like the transfer function for an LTI system.  However, we take a different path, as we mostly care about the steady-state limit for clustering. We first define the notion of weak equivalence between dynamical systems.
\begin{defn}
Two dynamical systems $\Upsilon_1$,$\Upsilon_2$ are called \emph{weakly equivalent} if their steady-state input-output relations are identical.
\end{defn}

\begin{exam} \label{exam.WeakEquiv}
Consider the following dynamical systems:
\small
\begin{align*}
\Upsilon_1&:\ y = u
&\Upsilon_2&: \begin{cases} \dot{x} = -x+u , \\ y = x  \end{cases}& \\
\Upsilon_3&: \begin{cases} \dot{x} = -10x+u , \\ y = 10x  \end{cases}
&\Upsilon_4&: \begin{cases} \dot{x} = -\tanh(x)+u , \\ y = \tanh(x)  \end{cases}& \\
\Upsilon_5&: \begin{cases} \dot{x} = -x+\sinh(u) , \\ y = \mathrm{arcsinh}(x)  \end{cases}
&\Upsilon_6&: \begin{cases} \dot{x} = -x+u , \\ y = 0.5(x+u)  \end{cases}&
\end{align*} \normalsize
These systems are vastly different from one another. One is memoryless, while the others are not. Some are LTI, and some are nonlinear. Of the nonlinear ones, one is input-affine nonlinear, while the other is not. All are output-strictly passive, but only $\Upsilon_6$ is input-strictly passive. However, all of these systems have the steady-state input-output relation $k(\mathrm u) = \mathrm u$, meaning that they are weakly equivalent. 
\end{exam}

\begin{defn}
Let $(\G,\Sigma,\Pi)$ be any multi-agent system for SISO agents. A \emph{weak automorphism} is a map $\psi:\mathbb{V}\to\mathbb{V}$ such that the following conditions hold:
\begin{itemize}
\item The map $\psi$ is an automorphism of the graph $\G$.
\item For any $i\in\mathbb{V}$, $\Sigma_i$ and $\Sigma_{\psi(i)}$ are weakly equivalent.
\item For any $e \in\mathbb{E}$, $\Pi_e$ and $\Pi_{\psi(e)}$ are weakly equivalent.
\item Moreover, if Assumption \ref{assump.3} does not hold, we demand that the map $\psi$ preserves edge orientation. 
\end{itemize}
We denote the collection of all weak automorphisms of the diffusively-coupled system $(\G,\Sigma,\Pi)$ by $\Aut(\G,\Sigma,\Pi)$. Naturally, this is a subgroup of the group of automorphisms $\Aut(\G)$ of the graph $\G$.
\end{defn}

The name ``weak" automorphism hints at the existence of a ``strong" automorphism. That would be an automorphism sending agents and controllers to agents and controllers having the same dynamics (e.g. that can be modeled using the same model). We shall not focus on that notion in this paper. 

\begin{notation}
Each permutation $\psi:\V\to\V$ defines a linear map $\mathbb{R}^{|\V|}\to\mathbb{R}^{|\V|}$ by permuting the coordinates according to $\psi$. We denote the linear operator by $P_\psi$. If $\psi$ is a graph automorphism for the graph $\G=(\V,\EE)$, then it gives rise to a permutation $\EE\to\EE$ on the edges. We denote the corresponding linear map $\mathbb{R}^{|\EE|}\to\mathbb{R}^{|\EE|}$ by $Q_\psi$. Namely, $(P_\psi)_{ij} = \delta_{ij}$ and $(Q_\psi)_{ef} = \delta_{ef}$ for $i,j\in \V$ and $e,f\in\EE$. 
\end{notation}

\begin{prop} \label{prop.Commutation}
For any graph $\G=(\V,\EE)$, and for any weak automorphism $\psi$, we have $P_\psi \E = \E Q_{\psi} D$ for some diagonal matrix $D$ with $\pm 1$ entries. Moreover, if $\psi$ preserves edge orientation, then $D = \mathrm{Id}_{|\V|}$ is the identity matrix. 
\end{prop}
\begin{proof}
Indeed, for every $i\in\V$ and $e\in\EE$,\small
\begin{align*}
[P_\psi\E]_{ie} &= \sum_{k\in\V} (P_\psi)_{ik}\E_{ke} = \sum_{k\in\V} \delta_{\psi(i)k}\E_{ke} = \E_{\psi(i),e} \\
[\E Q_{\psi}D]_{ie} &= \sum_{f\in\EE} \E_{if}(Q_{\psi})_{fe}D_{ee} =  \sum_{f\in\EE} \E_{if}\delta_{\psi(f)e}D_{ee} \\
&= \E_{i\psi^{-1}(e)}D_{ee}.
\end{align*}\normalsize
Thus, because $\psi(i) \in e$ if and only if $i \in \psi^{-1}(e)$, the entries of $\E$ are the same up to sign, which can be fixed by the matrix $D$. Moreover, if $\psi$ preserves edge orientations, then the signs are the same and $D=\mathrm{Id}_{|\V|}$.
\end{proof}

\subsection{Steady-State Clustering in Multi-Agent Systems} 

In this section, we wish to build a connection between the group action $\Aut(\G) \curvearrowright \G$ and the symmetries in the steady-state $\mathrm y$ of $(\G,\Sigma,\Pi)$. We start with the following proposition.
\begin{prop} \label{prop.Invariance}
The function $F(\mathrm y) = K^{\star}(\mathrm y) + \Gamma(\E^T\mathrm y)$ is $\Aut(\G,\Sigma,\Pi)$-invariant. In other words, $F(P_\psi \mathrm y) = F(\mathrm y)$ for any $y\in \mathbb{R}^{|\V|}$ and $\psi\in\Aut(\G,\Sigma,\Pi)$. 
\end{prop}
\begin{proof}
We first note that $K_i = K_{\psi(i)}$ and $\Gamma_e = \Gamma_{\psi(e)}$, as $k_i = k_{\psi(i)}$ and $\gamma_e = \gamma_{\psi(e)}$. Thus,
\begin{align*}
&K(P_\psi \mathrm y) = \sum_{i\in\V} K_i((P_\psi\mathrm y)_i) = \sum_{i\in\V} K_i(\mathrm y_{\psi(i)}) = \\ &\sum_{i\in\V} K_{\psi(i)}(\mathrm y_{\psi(i)}) =  \sum_{j\in\V} K_j((\mathrm y)_j) = K(\mathrm y),
\end{align*}
where we use the switch $j = \psi(i)$ and the fact that $\psi : V\to V$ is a bijection. Similarly, due to Proposition \ref{prop.Commutation}, one has
\begin{align*}
\Gamma(\E^T P_\psi \mathrm y) = \Gamma(D^TQ_\psi\E^T\mathrm y) = \sum_{e \in \EE} \Gamma_e(D_{ee} (Q_\psi\E^T\mathrm y)_e).
\end{align*}
If Assumption \ref{assump.3} holds, then $\Gamma_e(x) = \Gamma_e(-x)$, so because $D_{ee} \in \{\pm 1\}$, we can remove that term in the product. Otherwise, $D_{ee} = 1$. In any case, we get that the last expression is equal to \small
\begin{align*}
&\sum_{e\in \EE} \Gamma_e((Q_\psi\E^T\mathrm y)_e) = \sum_{e\in \EE } \Gamma_e((\E^T\mathrm y)_{\psi(e)}) = \\ &\sum_{e\in\EE} \Gamma_{\psi(e)}((Q_\psi\E^T\mathrm y)_{\psi(e)}) =  \sum_{k\in\EE} \Gamma_k((\E^T\mathrm y)_k) = \Gamma(\E^T\mathrm y),
\end{align*}\normalsize
where we switch $k=\psi(e)$. This completes the proof.
\end{proof}

\begin{cor}
Suppose that $(\G,\Sigma,\Pi)$ is a diffusively coupled network satisfying either Assumption \ref{assump.1} or \ref{assump.2}. Then the set of steady-state outputs for the diffusively coupled network $(\G,\Sigma,\Pi)$ is $\Aut(\G,\Sigma,\Pi)$-invariant, i.e., it is preserved when applying $P_\psi$-s for $\psi \in \Aut(\G,\Sigma,\Pi)$.
\end{cor}

\begin{proof}
Immediate from Theorem \ref{thm.Analysis} and Proposition \ref{prop.Invariance}.
\end{proof}

Up to now, we showed that if $\mathrm y$ is a possible steady-state output of the diffusively coupled network $(\G,\Sigma,\Pi)$, for some initial condition, then $P_\psi \mathrm y$ is also a possible steady-state output of the network, for some (maybe different) initial condition. We want to push the envelope and show that, actually, $P_\psi \mathrm y = \mathrm y$. Our main tool is strong convexity. 

\begin{thm} \label{thm.Symmetry}
Consider the diffusively-coupled system $(\G,\Sigma,\Pi)$, and suppose that either Assumption \ref{assump.1} or Assumption \ref{assump.2} hold. Then for any steady-state $\mathrm y$ of the closed-loop and any weak automorphism $\psi\in\Aut(\G,\Sigma,\Pi)$, $P_\psi \mathrm y = \mathrm y$.
\end{thm}

\begin{proof}
We recall that output-strictly MEIP systems have strictly monotone input-output steady-state relations \cite{Burger2014}, and that a convex function is strictly convex $\mathbb{R}\to\mathbb{R}$ if and only if its subdifferential is a strictly monotone relation \cite{Rockafeller1997}. Moreover, we recall that if $F$ is a strictly convex function defined on the affine subspace $\{x\in\mathbb{R}^n:\ Ax=b\}$ for some matrix $A$ and vector $b$, then it has a unique minimum \cite{Rockafeller1997}.

Suppose that Assumption \ref{assump.1} holds. Then $K_i$ are all strictly convex, and $\Gamma_e$ are all convex. Thus the function $F(\mathrm x) = K^\star(\mathrm x) + \Gamma(\E^T \mathrm x)$ is strictly convex, meaning it has a unique minimum, which is $\mathrm y$ by Theorem \ref{thm.Analysis}. Thus, because $P_\psi \mathrm y$ is also a minimizer of $F$, we conclude that $P_\psi \mathrm y = \mathrm y$.

Alternatively, suppose that Assumption \ref{assump.2} holds. In that case, the functions $K_i$ are convex and $\Gamma_e$ are strictly convex. Thus $F$ is strictly convex only in directions orthogonal to the consensus line $\mathrm{span}\{\mathbbm{1}_{|\V|}\}$, meaning that there could be more than one minimizer. However, we note that for any $d \in \mathbb{R}$, the function $F$ is strictly convex on the affine subspace $\mathcal{A}_d = \{x\in\mathbb{R}^{|\V|} | \mathbbm{1}_{|\V|}^Tx = d\}$, meaning that $F$ has a unique minimizer on each of these affine subspaces. Choose $d = \mathrm y^T \mathbbm{1}_{|\V|}$, so that $\mathrm y \in \mathcal{A}_d$.  Because $\mathrm y$ is a minimizer of $F$ on all of $\mathbb{R}^{|\V|}$, its also a minimizer on $\mathcal{A}_d$, making it the unique minimizer of $F$ on $\mathcal{A}_d$. Noting that $P_\psi \mathrm y$ is also a minimizer of $F$, and that $\mathbbm{1}_{|\V|}^T \mathrm y = \mathbbm{1}_{|\V|} P_\psi \mathrm{y}$, we get that $\mathrm y = P_\psi \mathrm y$.
\end{proof}

The theorem shows that the system converges to a steady-state $\mathrm y$ invariant under weak automorphisms. We want to restate it in a manner emphasizing the clustering that occurs. For that, we define the notion of exchangeability
\begin{defn}
We say that two agents $i,j\in \V$ are \emph{exchangeable} if there exists a weak automorphism $\psi \in \Aut(\G,\Sigma,\Pi)$ such that $\psi(i) = j$. We define the \emph{exchangeability graph} of the diffusively-coupled system $(\G,\Sigma,\Pi)$ as the graph $\HH = \HH(\G,\Sigma,\Pi) = (\V,\EE_\HH)$, where there is an edge $\{i,j\} \in \EE_\HH$ if $i$ and $j$ are exchangeable.
\end{defn} 

\begin{prop} \label{prop.EquivalenceRelation}
The exchangeability graph $\HH = \HH(\G,\Sigma,\Pi)$ is a union of disjoint cliques.
\end{prop}
\begin{proof}
It's enough to show that if there is a path between vertices $i,j$, then there is an edge $\{i,j\}$. Let $i,j$ be any two vertices, and suppose that there is a path $i=v_0, v_1,v_2,...,v_{k-1},v_k = j$ in $\HH$. By definition, there are weak automorphisms $\psi_0,...,\psi_{k-1}$ such that $v_r = \psi_{r-1}(v_{r-1})$ for any $r=1,2,...,k$. Because $\Aut(\G,\Sigma,\Pi)$ is a group, the composed map $\psi_{k-1}\psi_{k-2}\cdots\psi_1\psi_0$ is also a weak automorphism, and naturally, it maps $i$ to $j$. Thus the edge $\{i,j\}$ exists in the graph $\HH$, completing the proof.
\end{proof}

\begin{exam}\label{exam.Exchangeability}
Consider the graph $\G$ in Fig. \ref{fig.SynthesisGraph}, where nodes $1,3,4,5$ are all LTI with transfer function $G(s) = \frac{1}{s+1}$, and node $2$ is LTI with transfer function $G(s) = \frac{s}{2s+1}$. All edge controllers are static, having the form $\mu_e = \zeta_e$. We compute the exchangeability graph $\HH$ of the diffusively coupled system.
Suppose $\psi$ is an automorphism of $\G$. Then $\psi$ preserves the degree of each vertex. Thus, the sets $\{1,2\}$ and $\{3,4,5\}$ are all invariant under $\psi$. Moreover, $\psi$ cannot map $1$ to $2$, or vice versa, as the agents are not weakly equivalent. Furthermore, the map $\psi$ mapping $1\to 1, 2 \to 2$, and $3\to 4\to 5\to 3$ is a weak automorphism of the diffusively-coupled system. Thus the exchangeability graph $\HH$ contains the edges $\{3,4\},\{4,5\}$ and $\{5,3\}$. As we showed that agents $1$ and $2$ must remain invariant under weak automorphism, no more edges exist in the exchangeability graph $\HH$, so it is the union of three cliques - $\{1\},\{2\}$ and $\{3,4,5\}$. The graph can be seen in Fig. \ref{fig.ExchangeabilityExample}.
\end{exam}

\begin{figure}[b!]

\begin{center}
	\subfigure[Underlying Graph $\G$] {\scalebox{.33}{\includegraphics{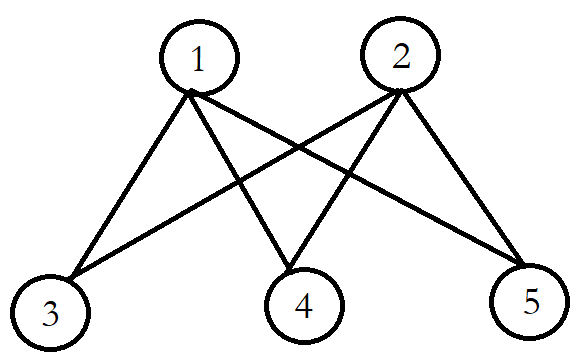}}\label{fig.SynthesisGraph}}
	\subfigure[The Exchangability Graph $\HH$]{\scalebox{.45}{\includegraphics{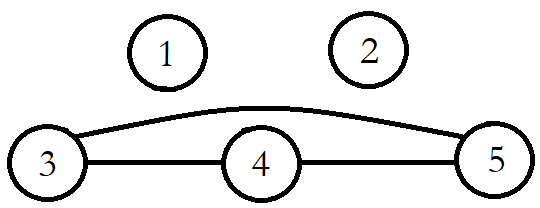}}\label{fig.ExchangeabilityExample}}
  \caption{Graphs for Example \ref{exam.Exchangeability}}
\end{center}

\end{figure}

We can now restate Theorem \ref{thm.Symmetry} in a more profound way:

\begin{thm}\label{thm.3}
Consider the diffusively-coupled system $(\G,\Sigma,\Pi)$, and suppose that either Assumption \ref{assump.1} or Assumption \ref{assump.2} hold. Then the system converges to a clustering steady-state, with clusters corresponding to the connected components of the exchangeability graph $\HH(\G,\Sigma,\Pi)$.
\end{thm}

\begin{proof}
The system converges to some steady-state $\mathrm y$ by Theorem \ref{thm.Analysis}. By Theorem \ref{thm.Symmetry}, The steady-state $\mathrm y$ is invariant to all weak-automorphisms. If we take any two vertices $\{i,j\}$ lying in the same connected component of the exchangeability graph $\HH$, then by Proposition \ref{prop.EquivalenceRelation}, the edge $\{i,j\}$ is in $\HH$. Thus there is an automorphism $\psi$ such that $\psi(i) = j$. Looking at the components of the equation $P_\psi \mathrm y = \mathrm y$ implies that $\mathrm y_i = \mathrm y_j$. In other words, we showed that the diffusively coupled system $(\G,\Sigma,\Pi)$ converges to a steady-state, and agents connected in the  exchangeability graph $\HH$ converge to the same limit. This completes the proof of the first part.
\end{proof}
\subsection{Homogeneous Networks and Cluster Synthesis} \label{subsec.Synthesis}
In many practical examples, we are dealing with a diffusively coupled network in which the agents are identical. Examples include neural networks, platooning, coupled oscillators, and other homogeneous swarms. Furthermore, in many practical scenarios we may desire to have all controllers in the system identical. This is the case where the agents have no identifiers like serial numbers. We can also try and use this frame to make clustering more robust - even if we use a wrong model for the controllers or the agents, we will still have clustering do to symmetry. It should be noted that designing controllers that force the system to cluster can be done by using the synthesis procedure appearing in \cite{Sharf2017,Sharf2018a}, but there is no guarantee that the achieved controllers will be (even weakly) homogeneous. We note that the built scheme allows us to tweak the notion of homogeneity:

\begin{defn}
A diffusively-coupled system is \emph{weakly homogeneous} if any two agents, and any two controllers, are weakly equivalent.
\end{defn}

As seen in Example \ref{exam.WeakEquiv}, weakly homogeneous systems can describe agents and controllers of many different kinds. Moreover, the notion of homogeneous networks allows us to study clustering using purely graph-theoretic and combinatorial methods. Indeed, we claim that weak automorphisms for $(\G,\Sigma,\Pi)$ are just graph automorphisms of $\G$.
\begin{prop}
Let $(\G,\Sigma,\Pi)$ be any weakly homogeneous diffusively coupled network. A map $\psi:\V\to \V$ is a weak automorphism of the system if and only if $\psi \in \Aut(\G)$.
\end{prop}
\begin{proof}
Follows from the definition of a weak automorphism, and weak equivalence of agents and controllers.
\end{proof}

Assumption \ref{assump.3} has special significance for weakly homogeneous systems.
\begin{thm}
Suppose that $(\G,\Sigma,\Pi)$ is weakly homogeneous, and either Assumption \ref{assump.1} or Assumption \ref{assump.2} holds. If Assumption \ref{assump.3} holds, then the network converges to consensus.
\end{thm}

\begin{proof}
Assumption \ref{assump.3} implies that $\gamma(0) = 0$, meaning that $\Gamma$ is minimized at $0$. If we let $\beta$ be $K_i=K_j$'s minimum.then $\mathrm y = \beta\mathbbm{1}_{|\V|}$ minimizes both $K(x)$ and $\Gamma(\E^Tx)$. Thus it minimizes (OPP), which is strictly convex in any direction orthogonal to the consensus line, completing the proof.
\end{proof}

Thus, as long as Assumption \ref{assump.3} does not hold, so consensus is not forced, clustering in weakly homogeneous systems can be understood in terms of the action of $\Aut(\G)$ on the graph $\G$. One interesting problem that can benefit from this framework is cluster synthesis. Namely, given fixed homogeneous agents, how can one design the interaction graph $\G$ and homogeneous controllers in order to achieve clustering with prescribed cluster sizes, at prescribed values. We give an example of a cluster synthesis problem in Section \ref{sec.Simulations}. However, it should be noted that the solution to these problems is not unique. For example, both the complete graph and cycle graph work when we want a single cluster, and there are many other solutions, such as Cayley graphs on finite groups \cite{Cayley1878}.

\section{Case Studies}\label{sec.Simulations}

\subsection{A Weakly Homogeneous Non-Homogeneous Network}
We consider a cycle graph $\G$ on $5$ nodes. The agents' models are given by $\Upsilon_2,\Upsilon_3,\Upsilon_4,\Upsilon_5,\Upsilon_6$ of Example \ref{exam.WeakEquiv}, where we add an identical random \emph{constant} exogenous input to all agents to avoid the mundane case of convergence to $\mathrm y = 0$. All the controllers on the edges are modeled as $\Upsilon_1$ of the same example. Obviously, this is a non-homogeneous weakly homogeneous network. Furthermore, the automorphism group $\Aut(G)$ can map any vertex in $\G$ to any other vertex, meaning that Theorem \ref{thm.3} implies that the system should converge to consensus. The output $y(t)$ and the relative output $\zeta(t)$ of the closed-loop system can be seen in Figure \ref{Fig.WeaklyHomogenous}. It is evident that the system indeed converges to consensus, up to numerical errors due to limited precision.

\begin{figure} [!t] 
    \centering

    \includegraphics[scale=0.3]{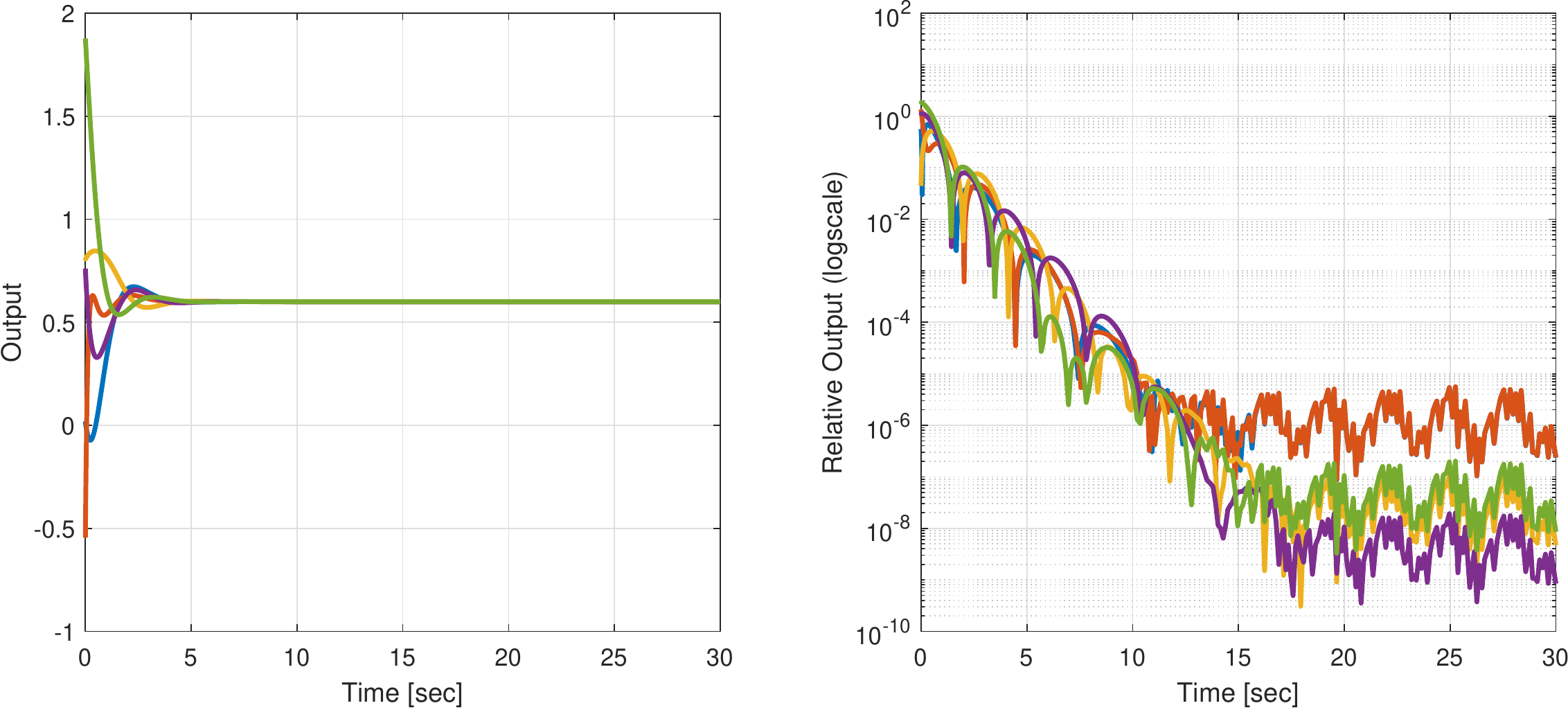}
    
    \caption{Output $y(t)$ and Relative Output $\zeta(t)$ for a Weakly Homogeneous Non-Homogeneous Network}
    \label{Fig.WeaklyHomogenous}
    
\end{figure}

\subsection{Cluster Synthesis - an Example}
We are given five agents, all are LTI with the TF $G(s) = \frac{1}{s+1}$. We wish to find a $\G$, and a collection of homogeneous controllers, such that $(\G,\Sigma,\Pi)$ converges to two clusters, one with two agents and one with three agents. The first cluster should be located at $\mathrm y=1$, and the second at $\mathrm y =0$.

First, according to the Discussion at Subsection \ref{subsec.Synthesis}, we want to find a graph $\G$, having five vertices, so that vertices $1,2$ are exchangeable, vertices $3,4,5$ are exchangeable. We consider the graph in Figure \ref{fig.SynthesisGraph}. Obviously, vertices $1,2$ are exchangeable, and vertices $3,4,5$ are exchangeable as well (but not with $1$ and $2$). It can be shown that this is the graph having the minimal number of edges possessing this property. We orient the edges from $1,2$ to $3,4,5$.

Now for the controller synthesis procedure. As we know from \cite{Sharf2017}, not all vectors are available as steady-state outputs of a diffusively-coupled network with prescribed agents. This can be fixed by allowing an addition of an (identical) constant exogenous input to all agents \cite{Sharf2018a}, and the steady-state equation becomes $\mathrm w = k^{-1}(\mathrm y) + \E\gamma(\E^T \mathrm y)$ \cite{Sharf2018b}. Writing this equation in coordinates, we get two equations, one for vertices in the $1^\mathrm{st}$ cluster, having $\mathrm y_i = 0$, and another for vertices in the $2^\mathrm{nd}$ cluster, having $\mathrm y_ i = 1$. We get:\small
\begin{align*}
w = 0 - 3\gamma_1(1-0) = -3\gamma_1(1) \\
w = 1 + 2\gamma_1(1-0) = 1+2\gamma_1(1) 
\end{align*}\normalsize

where $\gamma_1$ is the steady-state input-output relation for each copy of the homogeneous controller. We recall that we also need monotonicity, so $\gamma_1$ must be monotone. One possible solution to this set of equations is $w = 0.6$ and $\gamma_1(x) = -1.2+x$, the latter realized by the controller $\mu_e = -1.2 + \zeta_e$. We simulate the closed-loop system with the prescribed agents and synthesized graph and controllers. The output of the system is available in Figure \ref{fig.SynthesisSimulation}. It is evident that our solution indeed solves the cluster synthesis problem.

\begin{figure} [!bt] 
    \centering

    \includegraphics[scale=0.45]{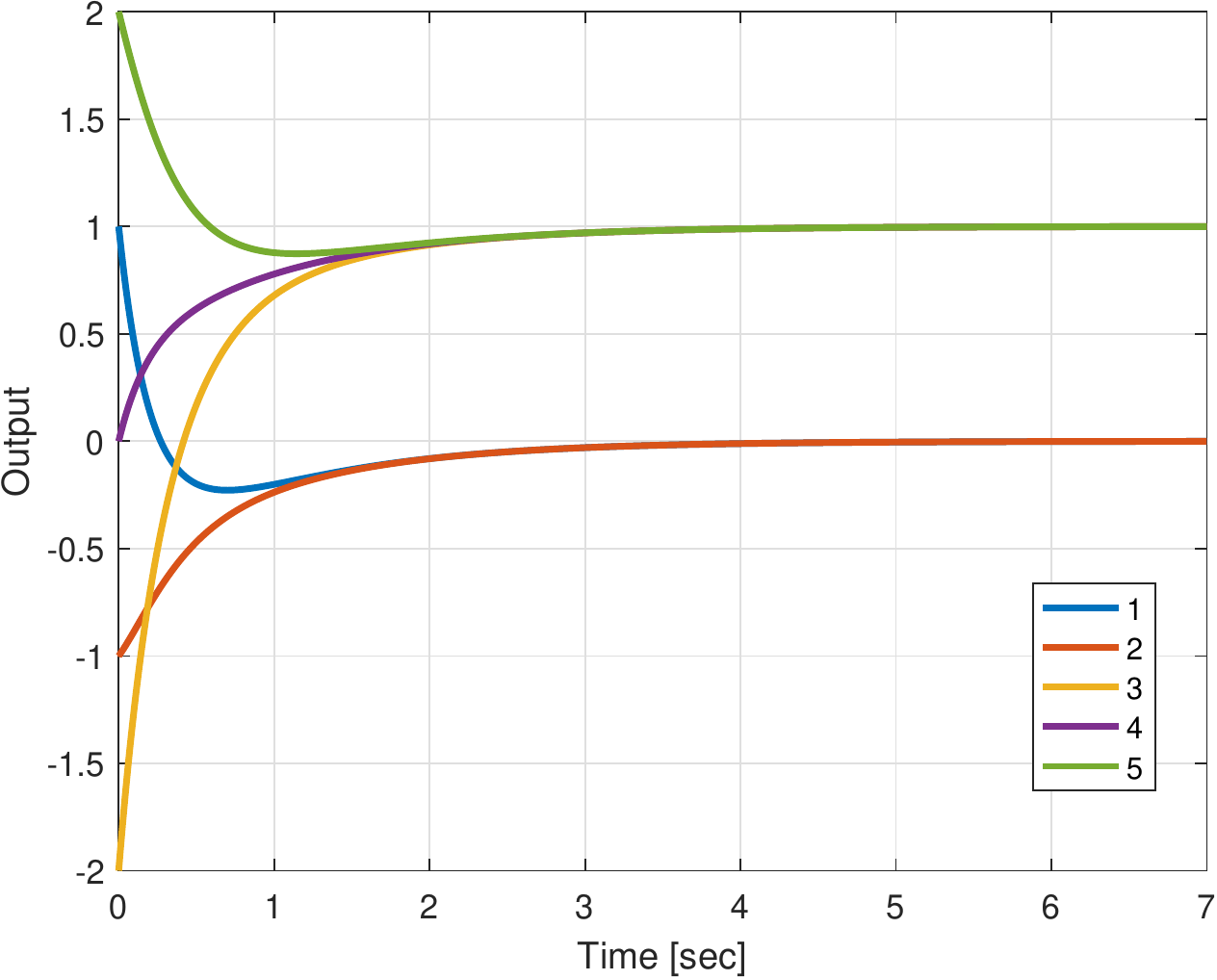}
    \caption{Cluster Synthesis - The Closed-Loop System.}
    \label{fig.SynthesisSimulation}
\end{figure} 
\section{Conclusion}

We presented the notion of weakly equivalent systems, and defined the weak automorphism group of a diffusively-coupled network. We showed that the set of all steady-state outputs of a diffusively-coupled network is invariant under weak automorphisms. Applied to networks of MEIP systems, we showed that networks satisfying either Assumption \ref{assump.1} or \ref{assump.2} must converge to a clustered steady-state output, with clusters corresponding to the connected components of the exchangeability graph. Later, we focused on weakly homogeneous networks, which are diffusively-coupled networks comprised of weakly equivalent agents and weakly equivalent controllers. We showed that the weak automorphism group of these systems is exactly the automorphism group of the underlying graph $\G$, and showed that if a weakly homogeneous network satisfies Assumption \ref{assump.3}, it converges to consensus. We discussed a possible application in synthesis of clusters in homogeneous networks. Lastly, we demonstrated the results in two different cases. Future research might seek a relaxed condition, requiring an even weaker notion of symmetry, as well as tackling the problem of cluster synthesis for homogeneous systems.

\bibliographystyle{ieeetr}
\bibliography{main}

\end{document}